\documentclass[twocolumn,a4paper,10pt]{article}

\usepackage{arw16}
\usepackage{graphicx}
\usepackage{amsthm,amsmath,amssymb} 
\usepackage{mdwlist}

 \renewenvironment{thebibliography}[1]{%
   \begin{odlthebibliography}{#1}%
     \setlength{\parskip}{0ex}%
     \setlength{\itemsep}{3pt}%
     \fontsize{9}{9} 
     \selectfont
}%
 {%
   \end{odlthebibliography}%
 }

\newtheoremstyle{jamiestyle}
  {4pt}
  {0pt}
  {\it}
  {0pt}
  {\sc}
  {.}
  { }
  {}
\theoremstyle{jamiestyle}
\newtheorem{thrm}{Theorem}[section]

\newtheorem{lemm}[thrm]{Lemma}

\newtheoremstyle{jamienfstyle}
  {4pt}
  {0pt}
  {\normalfont}
  {0pt}
  {\sc}
  {.}
  { }
  {}
\theoremstyle{jamienfstyle}

\newtheorem{defn}[thrm]{Definition}

\newtheorem{rmrk}[thrm]{Remark}

\newcommand\rulefont[1]{\ensuremath{{\mathrm{\bf (#1)}}}}

\newcommand\f[1]{\mathit{#1}}

\newcommand\atoms{\mathbb A}
\newcommand\deffont[1]{{\bf #1}}
\newcommand\powerset{\f{pset}}

\newcommand\ment[0]{\mathrel{\vDash}}
\newcommand\mact[0]{{{\cdot}{\cdot}}}
\newcommand\act[0]{{\cdot}}
\newcommand\ssm{{{:}\text{=}}}

\newcommand\liff{\Leftrightarrow}

\newcommand\Forall[1]{\forall #1.}
\newcommand\Exists[1]{\exists #1.}

\let\myfresh\#
\def\#{\ensuremath{\text{\tt\myfresh}}}

\newcommand\mnat\nat

\newcommand\seqpow\powerset
\newcommand\seqpowi\powerset

\newcommand\fv{\f{fv}}

\newcommand\vara{\text{\tt a}}
\newcommand\varb{\text{\tt b}}
\newcommand\varc{\text{\tt c}}

\newcommand\mjg[1]{} 

\newcommand\nat{{\mathbb N}}

\newcommand\finsubseteq{\mathbin{\subseteq_{\text{\it fin}}}}

\renewcommand\land{\wedge}
\renewcommand\lor{\vee}
\newcommand\limp{\Rightarrow}

\begin{document}

\title{Equivariant ZFA with Choice: a position paper}

\institute{Heriot-Watt University, Scotland, UK
}

\author{
	Murdoch J. Gabbay  
}

\abstract{
We propose \emph{Equivariant ZFA with Choice} as a foundation for nominal techniques that is stronger than ZFC and weaker than FM, and why this may be particularly helpful in the context of automated reasoning.
}

\maketitle

\section{Introduction}

Nominal techniques assume a set $a,b,c,\ldots\in\atoms$ of \emph{atoms}; elements that can be compared for equality but which have few if any other properties. 
This deceptively simple foundational assumption has many applications---nominal abstract syntax (syntax-with-binding)~\cite{gabbay:newaas-jv,pitts:nomsns}; as implemented in Isabelle~\cite{urban:nomrti}; an open  consistency problem~\cite{gabbay:conqnf}; duality results~\cite{gabbay:stodfo,gabbay:repdul}; generalised finiteness for automata and regular languages~\cite{kozen:comink,bojanczyk:auttns}; rewriting with binding~\cite{gabbay:nomr-jv}; and more.

So what is a foundation for nominal techniques? 

Where this question is addressed in the nominal literature, the answer given is \emph{Fraenkel-Mostowski set theory} (\deffont{FM}). 
In this position paper I will somewhat provocatively outline why this may have been a mistake, or at least a suboptimal choice.
I will propose \emph{Equivariant ZFA set theory with Choice} (\deffont{EZFAC}) instead, and suggest why EZFAC may be especially suited to applications in automated reasoning and implementation.
One standout point is that FM is inconsistent with the Axiom of Choice, whereas EZFA plus Choice (EZFAC) is consistent.

An expanded discussion of EZFAC is in~\cite{gabbay:equzfn}.

\section{Equivariant ZFA with Choice}

\begin{defn}
The language of EZFAC is the language of \emph{sets with atoms}---first-order logic with a binary predicate $\in$, and a single constant symbol $\mathbb A$ for the \deffont{set of atoms}.\footnote{An Equivariant Higher-Order Logic with Atoms and Choice would be equally feasible.  I discuss sets rather than simple types only for convenience.}
Axioms are in Figure~\ref{fig.zfa}; notation is defined below.
\end{defn}

\begin{rmrk}
Axioms \rulefont{AtmEmp} to \rulefont{Choice} are standard ZFAC (ZF with atoms and Choice).
In rule \rulefont{Choice}, $\powerset^*$ is the \emph{nonempty powerset} operator, since we cannot choose an element of the empty set.
\end{rmrk}

\begin{rmrk}
ZF and ZFA are equally expressive: a model of ZFA embeds in one of ZF,\footnote{\dots by modelling atoms as $\mathbb N$, or $\powerset(\mathbb N)$, and so forth.} and vice-versa, and a predicate in ZFA can be translated (quite easily) to one in ZF.

Yet if the translation from ZFA to ZF leads to a quadratic increase in proof-size, or if ZFA is an environment which naturally lets us express \emph{native ZFA concepts};\footnote{\dots meaning concepts that are hard to address in full generality in ZF, where we do not have atoms, but easy to see in ZFA, where we do.} 
then the gain from ZFA can be useful, as we will consider. 
\end{rmrk}

\begin{figure}
\scalebox{.8}{$
\begin{array}{l@{\quad}l@{\hspace{-6em}}l}
\rulefont{AtmEmp}&
t\in s\limp s\not\in\atoms
\\
\rulefont{EmptySet}&
t\not\in\varnothing
\\
\rulefont{Ext}&
s,s'\not\in\atoms\limp (\Forall{\varb}(\varb\in s\liff \varb\in s'))\limp s=s'
\\
\rulefont{Pair}&
t\in\{s,s'\} \liff (t=s\lor t=s')
\\
\rulefont{Union}&
t\in\bigcup s \liff \Exists{\vara}(t\in\vara\land \vara\in s)
\\
\rulefont{Pow}&
t\in\powerset(s)\liff t\subseteq s 
\\
\rulefont{Ind}&
(\Forall{\vara}(\Forall{\varb{\in}\vara}\phi[\vara\ssm\varb])\limp\phi) \limp \Forall{\vara}\phi
&\fv(\phi)=\{\vara\}
\\
\rulefont{Inf}&
\Exists{\varc}\varnothing\in\varc\land\Forall{\vara}\vara\in\varc\limp\vara{\cup}\{\vara\}\in\varc
\\
\rulefont{AtmInf}&
\neg(\atoms\finsubseteq\atoms) 
\\
\rulefont{Replace}&
\Exists{\varb}\Forall{\vara}\vara\in\varb\liff \Exists{\vara'}\vara'\in u\land \vara=F(\vara')
\\
\rulefont{Choice}&
\varnothing\neq (\powerset^*(s)\to s) \qquad\powerset^*\text{ nonempty powerset}
\\ 
\rulefont{Equivar} & 
\Forall{\vara{\in}\f{Perm}}(\phi\liff\vara\mact\phi) .
\end{array}
$
}
\caption{Axioms of EZFAC}
\label{fig.zfa}
\end{figure}

\begin{defn}
A \deffont{permutation} $\pi$ is a bijection on $\atoms$.
Define \deffont{permutation action} $\pi\act\vara$ by: $\pi\act\vara=\pi(\vara)$ if $\vara\in\atoms$ and $\pi\act\vara=\{\pi\act\varb\mid \varb\in\vara\}$ if $\vara\not\in\atoms$.

Then given a predicate $\phi$ in the language of set theory with atoms, define $\pi\mact\phi$ to be that predicate obtained by replacing every free variable $\vara$ with $\pi\act\vara$.
\end{defn}

\begin{rmrk}
\label{rmrk.read.equivar}
\rulefont{Equivar} asserts that validity is preserved by permuting atoms in all parameters of a predicate.
For instance if we have proved $\phi(a,b,c)$ for atoms $a,b,c\in\atoms$ then 
taking $\pi=(a\,c)$ we have $\phi(c,b,a)$ and 
taking $\pi=(a\,a')(b\,b')(c\,c')$ we have $\phi(a',b',c')$.
We do not have $\phi(a,b,a)$; this may still hold, but not by \rulefont{Equivar} because no permutation takes $(a,b,c)$ to $(a,b,a)$.

This gives atoms a dual nature.
Individually atoms point to themselves,\footnote{In the Isabelle implementation of FM in my PhD thesis~\cite{gabbay:thesis} this was literally so: I used \emph{Quine atoms} such that $a=\{a\}$.  This removes the condition $a,b\not\in\atoms$ in \rulefont{Ext}, at a cost of some extremely mild non-wellfoundedness.} 
but collectively atoms have the flavour of variables ranging permutatively over $\atoms$.\footnote{To see this made precise see Subsection~2.6 and Lemma~4.17 of \cite{gabbay:pnlthf}.} 
\end{rmrk} 

\begin{rmrk}
\label{rmrk.equivar.native}
{\bf Equivariance is a native ZFA concept}.
If our intuitions are ZF-shaped, then equivariance seems counterintuitive: ``We can't just permute elements.  Suppose atoms are numbers: then are you claiming $1<2$ if and only if $2<1$?''.
ZFA makes clear what is going on: the premise ``Suppose atoms are numbers'' makes no sense, because atoms are \emph{not} numbers!
\end{rmrk}

\section{Equivariance, choice, and freshness}

$\pi$ acts bijectively \dots 
\begin{lemm}
\label{lemm.base.equivar}
Suppose $\mathfrak M$ is a model of ZFA(C).
Then
$\mathfrak M\ment \pi\act y\in\pi\act x$ if and only if $\mathfrak M\ment y\in x$.
\end{lemm}
\dots so a model of ZFA or ZFAC \emph{already} satisfies \rulefont{Equivar} \cite[Theorem~8.1.10]{gabbay:thesis}:
\begin{thrm}
\label{thrm.equivar.for.free}
\label{thrm.equivar}
If $\mathfrak M$ is a model of ZFA(C) then $\mathfrak M$ is also a model of EZFA(C).
\end{thrm}
\begin{proof}
We prove $\mathfrak M\ment\phi\liff\pi\mact\phi$ by induction on $\phi$.
The case of $t\in s$ is Lemma~\ref{lemm.base.equivar}.
The cases of $\Forall{\vara}\phi$ and $\atoms$ use the fact that $\pi$ is bijective.
Other cases are no harder.
\end{proof}

\begin{rmrk}
\label{rmrk.natural.axiom}
{\bf Why do we need \rulefont{Equivar}?}
While Theorem~\ref{thrm.equivar} shows how instances \rulefont{Equivar} can be derived in ZFA, in practice the cost of proving them from first principles \`a la Theorem~\ref{thrm.equivar} scales with the complexity of the predicate $\phi$.
Instances of Theorem~\ref{thrm.equivar} can quadratically dominate development effort in a theorem-prover.\footnote{This happened in the Isabelle/FM implementation in my thesis~\cite{gabbay:thesis}, and it was crippling. After my PhD I initiated a mark~2 development with an \rulefont{Equivar} axiom-scheme (actually an Isabelle Oracle).} 
This is the problem of $\alpha$-equivalence, come back to bite us.
In contrast, axiom \rulefont{Equivar} costs \emph{constant} effort: namely, the cost of invoking the axiom.\footnote{The modern nominal Isabelle implementation uses automated tactics to prove Theorem~\ref{thrm.equivar} for certain classes of $\phi$, specialised to an application to nominal inductive datatypes.  Nominal Isabelle good at what it does, but it is not (and never claimed to be) a universal nominal foundation.} 
In this sense, equivariance is a \emph{natural axiom}.
\end{rmrk}

\begin{rmrk}
\label{rmrk.Choice.1}
{\bf Choice compatible with equivariance.}
Surely arbitrary choices are inherently non-equivariant?
Not if they are made inside $\mathfrak M$.
Consider some choice-function 
$f\in\powerset^*(x)\to x$.
By Theorem~\ref{thrm.equivar} we immediately obtain
$\pi\act f\in \powerset^*(\pi\act x)\to \pi\act x$.
In words: if $f$ is a choice function for $x$ in $\mathfrak M$ then by equivariance $\pi\act f$ is a choice function for $\pi\act x$ in $\mathfrak M$.
We just permute atoms pointwise in the choice functions.

So the following are consistent with EZFA and derivable in EZFAC:
``There exists a total ordering on $\mathbb A$'';
``Every set can be well-ordered (even if the set mention atoms)''.
\end{rmrk}

\begin{rmrk}
FM set theory has a \emph{finite support} property that for every $x$ there exists $A\subseteq_{fin}\atoms$ such that if $\Forall{a{\in}A}\pi(a)=a$ then $\pi\act x=x$.

I argue that it is better to 
present freshness as a \emph{well-behavedness} property in the larger EZFA(C) universe.
This is for several reasons:
\\
1. Presenting \rulefont{Fresh} as a well-behavedness property instead of an axiom eliminates the `But FM is inconsistent with Choice' objection to nominal techniques.
Anything we can do in FM, we can very easily do in EZFAC just by imposing finite support.  Choice functions need not have finite support, but (unlike is the case for FM) they still exist in the same EZFAC universe.
\\
2. We sometimes specifically want non-supported elements; for example two recent papers \cite{gabbay:stodfo,gabbay:repdul} are concerned with sets that have a notion of nominal support, but whose elements do not.
\\
3. Support is not a hereditary property; e.g. `the set of all well-orderings of atoms' is supported by $\varnothing$, but no well-ordering of $\atoms$ has finite support; or put another way, the FM universe is a proper subclass of the universe of finitely-supported elements. 
\\
4. Even if the reader's next paper uses FM sets, it may be helpful for exposition to observe the natural embedding of the FM universe inside the EZFAC universe.
Of course this embedding is obvious, but only \emph{once it is pointed out}. 
\end{rmrk}

\section{Conclusions}

Nominal techniques have developed considerably since the original work~\cite{gabbay:newaas-jv} yet their foundations have not been critically revisited.
Authors and implementors have generally used FM or ZF(C), if foundations are explicitly considered.
Yet there is a sense in which FM is too strong, and ZF(C) is too weak.
Though these theories are all biinterpretable, that is not enough.
We need foundations that allow us express ourselves precisely, naturally, and without annoying (or crippling) proof-obligations to do with renamings.

In this respect EZFAC seems to have advantages. 
We have Choice, invoking equivariance has constant cost---and just the clear statement of EZFAC itself will I hope be conceptually useful. 

\newcommand\href[2]{#2}

\hyphenation{Mathe-ma-ti-sche}

\end{document}